\documentclass[12pt]{article}
\begin{document}

\setlength{\textheight}{8.8in}
\setlength{\textwidth}{6.5in}
\setlength{\evensidemargin}{-0.18in}
\setlength{\oddsidemargin}{-0.18in}
\setlength{\headheight}{0in}
\setlength{\headsep}{10pt}
\setlength{\topsep}{0in}
\setlength{\topmargin}{0.0in}
\setlength{\itemsep}{0in}      % 10pt is too big with the 1.2 stretch
\renewcommand{\baselinestretch}{1.2}
\parskip=0.080in

\newcommand{\qedsymb}{\hfill{\rule{2mm}{2mm}}}
\def\squarebox#1{\hbox to #1{\hfill\vbox to #1{\vfill}}}
\newcommand{\qed}{\hspace*{\fill}
        \vbox{\hrule\hbox{\vrule\squarebox{.667em}\vrule}\hrule}\smallskip}
\newenvironment{proof}{\begin{trivlist}
\item[\hspace{\labelsep}{\em\noindent Proof: }]
}{\qed\end{trivlist}}

\newtheorem{theorem}{Theorem}[section]
\newtheorem{lemma}[theorem]{Lemma}
\newtheorem{claim}[theorem]{Claim}
\newtheorem{corollary}[theorem]{Corollary}
\newtheorem{proposition}[theorem]{Proposition}
\newtheorem{observation}[theorem]{Observation}
\newtheorem{defin}[theorem]{Definition}
\newtheorem{definition}[theorem]{Definition}
\newtheorem{example}[theorem]{Example}
\newtheorem{conjecture}[theorem]{Conjecture}

\title{Edge Covering with Budget Constrains
}
\author{
Rajiv Gandhi\thanks{Department of Computer Science,
Rutgers University, Camden, NJ 08102.
Partially supported by NSF grant number 1218620
.
E-mail: {\bf rajivg@camden.rutgers.edu}}
\and
Guy Kortsarz\thanks{Department of Computer Science,
Rutgers University, Camden, NJ 08102.
Partially supported by NSF grant number 1218620.
E-mail: {\bf guyk@camden.rutgers.edu}.}}

% Pseudo-code line numbering

\newcommand{\lab}{\ell}
\newcommand{\iter}{\mathit{iter}}

%local commands for this paper
\newcommand{\kminL}{{\small \sf Fixed cost minimum edge cover }}
\newcommand{\mwmsL}{{\small \sf Maximum weight $m'$-edge cover }}
\newcommand{\kmin}{{\small \sf FCEC }}
\newcommand{\mwms}{{\small \sf MWEC }}
\newcommand{\dks}{{\small \sf Dense $k$-subgraph }}
\newcommand{\dap}{{\small \sf Degrees density augmentation }}

\maketitle

\begin{abstract}
We study two related problems: the \mwmsL (\mwms) problem and the
\kminL (\kmin)
problem. In the \mwms problem, 
we are given an undirected 
simple graph $G=(V,E)$ with integral vertex weights.
The goal is to select a set $U\subseteq V$ of maximum weight
so that the number of edges with at least one endpoint 
in $U$ is at most $m'$. 
Goldschmidt and Hochbaum \cite{GH97} show that the problem is NP-hard
and they give a $3$-approximation algorithm for the problem. 
We present an approximation algorithm that achieves a guarantee of
$2$, thereby improving the bound of $3$ \cite{GH97}.
In the \kmin problem, we are given a vertex weighted graph, a bound
$k$, and our goal is to find a subset of vertices $U$ of total weight
at least $k$ 
such that the number of edges with at least one 
edges in in $U$ is minimized.
A $2(1+\epsilon)$-approximation for the problem follows from the work of
Carnes and Shmoys \cite{CS08}.
We improve the approximation ratio by giving a $2$-approximation
algorithm for the problem.  
Can we get better results using methods based
on linear programming? We take a first step and show that 
the natural LP for \kmin has an integrality gap of $2-o(1)$.
We improve the NP-completeness result for \mwms \cite{GH97} by showing
that unless a well-known instance of the \dks admits a
constant ratio, \kmin and \mwms do not admit a PTAS. Note that the best
approximation guarantee known for this instance of \dks is
$O(n^{2/9})$ \cite{BCCFV}.
We show that for any constant $\rho>1$, an approximation guarantee of
$\rho$ for the \kmin problem implies 
a $\rho(1+o(1))$ approximation for \mwms.
Finally, we define
the \dap problem which is the density version of the \kmin problem.
In this problem we are given an undirected graph $G=(V,E)$ and 
a set $U\subseteq V$.
The objective is to 
find a set $W$ so that $(e(W)+e(U,W))/deg(W)$ is maximum.
This problem admits an LP-based exact solution \cite{CMNRS}. 
We give a combinatorial algorithm for this problem. 
\end{abstract}

\section{Introduction}
\label{sec:intro}
We study two natural budgeted edge covering problems in undirected
simple graphs with integral weights on vertices.
The budget is given either as a bound on the number of edges to be
covered or as a bound on the total
weight of the vertices. We say that an edge $e$ is \textit{touched} by a set of
vertices $U$ or that $e$ \textit{touches} the set of vertices $U$, if at least
one of $e$'s endpoints is in $U$.
Specifically, the problems that we study are as follows. 
The \mwmsL (\mwms) problem that we study was first introduced by Goldschmidt
and Hochbaum \cite{GH97}. In this problem, we are given an undirected 
simple graph $G=(V,E)$ with integral vertex weights.
The goal is to select a subset $U\subseteq V$ of maximum weight
so that the number of edges touching
$U$ is at most $m'$. 
This problem is 
motivated by application in loading of semi-conductor components to be
assembled into products \cite{GH97}.

We also study the closely related \kminL (\kmin) problem in which
given a graph $G=(V,E)$ with vertex weights and 
a number $W$, our goal is to 
find $U\subseteq V$ of weight at least $W$ such that the number 
of edges touching $U$ is minimized.

Finally, we study 
the \dap problem which is the density version of the \kmin problem.
In the \dap problem, we are given an
undirected graph graph $G=(V,E)$ and 
a set $U\subseteq V$ and our goal is to  
find a set $W$ with maximum augmenting density 
i.e., a set $W$ that maximizes $(e(W)+e(U,W))/deg(W)$.

\subsection{Related Work}
Goldschmidt and Hochbaum \cite{GH97} introduced the \mwms problem. 
They show that the problem is NP-complete and give algorithms that yield
$2$-approximate and $3$-approximate algorithm for
the unweighted and the weighted versions of the problem,
respectively.

A class of related problems are the density problems -- problems in
which we are to find a subgraph and the objective function considers
the ratio of the total number or weight of edges in the subgraph
to the number of vertices in the subgraph. A well known problem in
this class is the \dks problem ($DkS$) in which we want to find a subset of
vertices $U$ of size $k$ such that the total number of edges
in the subgraph induced by $U$ is maximized.
%This problem is not well understood and is surprisingly hard.
The best ratio known for the problem is 
$n^{1/4+\epsilon}$ \cite{FKP,BCCFV}, which is an improvement over the
bound of $O(n^{1/3-\epsilon})$, for $\epsilon$ close to $1/60$
\cite{FKP}. The \dks problem is
APX-hard under the assumption that 
NP problems can not be solved in subexponential time
\cite{K06}. Interestingly, if there is no bound on the the size of
$U$ then the problem can be solved in polynomial time \cite{Lawler,G84}.

Consider an objective function in which we minimize $deg(U)$. 
One can associate a cost $c_u=deg(u)$ with each vertex $u$ and
a size $s_u=w(u)$ for each vertex $u$, and then the objective is just
to minimize $deg(U)$ subject to $\sum s_ux_u \geq k$. Carnes
and Shmoys \cite{CS08} give a $(1+\epsilon)$-approximation for the
problem. Using this result and the observation that the 
objective function is at most a factor of $2$ away from the objective
function for the \kmin problem, a $2(1+\epsilon)$-approximation
follows for the \kmin problem.

Variations of the \dks problem in which the size of $U$ is at least
$k$ ($Dalk$) and the size of $U$ is at most $k$ ($Damk$) have been studied
\cite{AC09,KS09}. In \cite{AC09, KS09}, they give evidence that $Damk$
is just as hard as $DkS$. They also give $2$-approximate solutions to
the $Dalk$ problem. In \cite{KS09}, they also consider the density
versions of the problems in directed graphs. Gajewar and Sarma \cite{GS12}
consider a generalization in which we are give a partition of vertices
$U_1, U_2, \ldots, U_t$, and non-negative integers $r_1, r_2, \ldots,
r_t$. the goal is to find a densest subgraph such that partition $U_i$
contributes at least $r_i$ vertices to the densest subgraph. They give
a $3$-approximation for the problem, which was improved to $2$ by
Chakravarthy et al. \cite{CMNRS}, who also consider other
generalizations. They also show using linear programming that the \dap
problem can be solved optimally.
 
A problem parameterized by $k$ 
is Fixed Parameter Tractable \cite{dany}, if it 
admits an exact algorithm with running time of $f(k)\cdot n^{O(1)}$.
The function $f$ can be exponential in $k$ or larger.
Proving that a problem is W[1]-hard (with respect to parameter $k$)
is a strong indication that it has no FPT algorithm with parameter $k$
(similar to NP-hardness implying the likelihood of no polynomial time algorithm).
The {\kmin} problem parameterized by $k$ is W[1] hard 
but admits a $f(k,\epsilon)\cdot n^{O(1)}$ time,
$(1+\epsilon)$-approximation, for any constant $\epsilon>0$ \cite{dany}.
This is in contrast to our result that shows that it is highly
unlikely that \kmin admits
a polynomial time approximation scheme (PTAS), if the running time is bounded by a polynomial in $k$.

\subsection{Preliminaries}
The input is an undirected simple graph $G=(V,E)$ and vertex weights
are given by $w(\cdot)$. Let $n=|V|$ and $m=|E|$.
For any subset $S\subseteq V$, let $\overline{S} = V\setminus S$. Let
$e(P,Q)$ be the set of edges with one endpoint in $P$ and the other in
$Q$. Let $deg(S)$ denote the sum of degrees of all vertices in $S$,
i.e., $deg(S)=\sum_{v\in S}deg(v)$. Let $deg_H(v)$ denote the number
of neighbors of $v$ among the vertices in $H$. Let $deg_H(S)$ denote
the quantity $\sum_{v\in S}deg_H(v)$. We use $OPT$ to denote an
optimal solution as well as the cost of an optimal solution. The
meaning will be clear from the context in which it is used.

For set $U\subseteq V$, let  
$T(U)$ be the collection of all edges with at least one endpoint in $U$.
Namely, is the set of edges touching $U$.
We denote $t(U)=|T(U)|$.
The set of edges with both endpoints in $U$, also called {\em internal} edges 
of $U$, is denoted by $E(U)$.
We denote $e(U)=|E(U)|$.
We denote by $e(X,Y)$ the number of edges with one endpoint in 
$X$ and one in $Y$. Let $e_{U}(X,Y)$ be the number of edges 
between $X\cap U$ and $Y\cap U$ in the graph $G(U)$ induced 
by $U$.

\begin{lemma} 
\label{lemma:klowest}
The \kmin problem admits a simple $2$-approximate solution in case of
uniform vertex weights.
\end{lemma}

\begin{proof}
Let $Z$ be the set of $k$ lowest degree vertices in $G$. The set $Z$
is a $2$-approximate solution. Why? Let $b$ be the average degree of
vertices in $Z$. Thus $t(Z) \leq bk$. The claim follows since $t(OPT) \geq
deg(OPT)/2 \geq bk/2$.  
\end{proof}

From Lemma \ref{lemma:klowest}, if $deg(OPT)\geq bk(1+\epsilon)$, we
obtain a $2/(1+\epsilon)<2$ approximation guarantee using the set $Z$
as the solution.
Henceforth we assume that $deg(OPT)<bk(1+\epsilon)$

\begin{claim}
For every set $U$, $t(U)=deg(U)-e(U)$
\end{claim}

\begin{proof}
Consider separately the edges 
$E(U,V\setminus U)$ and $E(U)$.
Note that the edges $E(U,V\setminus U)$ are counted once 
in the sum of degrees, but edges in $E(U)$ are counted twice.
Thus in order to get the number of edges touching $U$,
we need to subtract $e(U)$ from $deg(U)$. 
\end{proof}

\subsection{Our results}

Our contributions in this paper are as follows.
\begin{itemize}
\item
For the \mwms problem we
give an algorithm that yields an approximation guarantee of $2$, thereby improving
the guarantee of $3$ given by Goldschmidt and Hochbaum \cite{GH97}. 

\item
We give a $2$-approximate solution to the 
\kmin problem. This improves the $2(1+\epsilon)$-ratio that follows
from the work of Carnes and Shmoys \cite{CS08}.

\item
Can linear programming 
be used to improve the ratio of $2$ for \kmin and \mwms problems?
We take a first step and show that a natural LP for 
{\kmin} has an integrality gap of $2(1-o(1))$, even for the unweighted case.

\item
We show that unless a well-known instance of the \dks admits a
constant ratio, \kmin and \mwms do not admit a PTAS. Note that the best
approximation guarantee known for this instance of \dks is
$O(n^{2/9})$ \cite{BCCFV}. This gives a stronger hardness result than
the NP-completeness result known for \mwms \cite{GH97}. 

\item
For any constant $\rho>1$, we show that if {\kmin} admits a
$\rho$-approximation algorithm then 
{\mwms} admits a $\rho(1+o(1))$-approximation algorithm. 

\item
We give a combinatorial algorithm that solves the {\dap}
problem optimally. 

\end{itemize}

\section{A 2-approximation for Maximum Weight
  $m'$-Edge Cover}
In this section we give a dynamic programming based solution for
the \mwms problem. The idea of using dynamic programming in this
context was first proposed by Goldschimdt and Hochbaum \cite{GH97}.
Recall that in the \mwms problem, we are 
given an undirected simple graph $G=(V,E)$ with integral vertex weights.
The goal is to select a subset $U\subseteq V$ of maximum weight
so that the number of edges touching
$U$ is at most $m'$.
 
We will guess the following entities (by trying all possibilities) and
for each guess, we use dynamic programming to solve the
problem. 
\begin{enumerate}
\item
$H^* = \{v_h\}$, where $v_h$ is the heaviest vertex in an optimal solution.
\item
$P_{H^*} = e(H^*, OPT\setminus H^*)$ -- the number of neighbors of $v_h$ in the
  optimal solution. There are at most $n$ possibilities.
\item
$D_{H^*} = deg_{\overline{H}^*}(OPT\setminus H^*)$: total degree of
  vertices in $OPT\setminus H^*$ in the graph induced by vertices in
  $V\setminus H^*$. There are at most $n^2$ possibilities.
\end{enumerate}

We will try all combinations of the above entities. Since there are at
most polynomial number of possibilities for each entity, we have at
most polynomial number of possibilities in total.
We define the following subproblems as part of our dynamic
programming solution.  Let $H$ be a guess for the singleton set $H^*$
that contains the heaviest vertex in an optimal solution. Let
$\{v_1,v_2, \ldots, v_{n-1}\}$ be the vertices in 
$\overline{H}$ (recall $\overline{H}=V\setminus H$). 
Then, for any $H$, we solve the following subproblems.
\begin{quote}
$A[H, i, P_{H}, D_{H}]$ denote the maximum weighted subset
$Q\subseteq \{v_1, v_2, \ldots, v_i\}$ such that $e(H,Q)\geq P_{H}$
and $deg_{\overline{H}}(Q)\leq D_{H}/2$.
\end{quote}
Note that while the natural bound on $deg_{\overline{H}}(Q)$ is $D_H$,
using such a bound will lead to an infeasible solution.
For fixed parameters $H$, $P_{H}$, and $D_{H}$, we are interested in
$A[H,n-1,P_H,D_H]$. We use the following recurrence as the basis for
our dynamic programming solution: the value of $A[H,i,P_{H},D_{H}]
= -\infty$ in any of the following three cases -- (i) $i=0$ and $P_{H}>0$, (ii) $i=0$
and $D_{H} < 0$, and (iii)  $D_{H}/2 > m'-e(H,\overline{H})$.
When $i=0$, $P_{H}\leq 0$ and $D_{H}\geq 0$, the value of
$A[H,i,P_{H},D_{H}] = 0$. Otherwise, we have   
\begin{eqnarray*}
A[H,i,P_{H},D_{H}] & = \max\{A[H,i-1,P_{H},D_{H}], w(v_i)+A[H,i-1,P'_{H},D'_{H}]\} 
\end{eqnarray*}
where, $P'_{H} = P_{H}-deg_{H}(v_i)$ and $D'_{H} = D_{H} - 2(deg_{\overline{H}}(v_i))$. 
Our solution is given by 
$\max_{H,P_{H},D_{H}}\{w(H)+A[H,n-1, P_{H}, D_{H}]\}$. 

\subsection*{Analysis}
\begin{lemma}
Our algorithm yields a feasible solution.
\end{lemma}

\begin{proof}
Let $H'\cup Q'$, where $Q'\subseteq V\setminus H'$, be the set of
vertices returned by our solution. The number of edges with at least
one endpoint in $H'\cup Q'$, is 
\begin{eqnarray*}
& = e(H', \overline{H}') + e(Q', \overline{H}') \\
& \leq  e(H', \overline{H}') +  deg_{\overline{H'}}(Q')\\
& \leq  e(H', \overline{H}') +  \frac{D_{H'}}{2}\\
& \leq e(H', \overline{H'}) +
  (m'-e(H',\overline{H'})) \hspace*{1cm} (\mbox{using the
    base case})\\
& = m'
\end{eqnarray*}
\end{proof}

\begin{lemma}
The above algorithm results in a $2$-approximate solution.
\end{lemma}

\begin{proof}
Recall that $H^*$ consists of the highest degree vertex in the
optimal solution. Let $Q^*$ be the remaining vertices in the
optimal solution. Consider the scenario when our algorithm makes the
correct guess for $H^*$.
Let $Q\subseteq \overline{H^*}$ be the solution returned by the dynamic
program in this setting. We know that 
\[
deg_{\overline{H}^*}(Q) \leq \frac{deg_{\overline{H}^*}(Q^*)}{2}
\]
We now use ideas from \cite{GH97} to show that $w(H^*\cup Q)\geq
2w(H^*\cup Q^*)$. Recall that $H'\cup Q'$ be the output of our
algorithm. Since $w(H'\cup Q')\geq w(H^*\cup Q^*)$,
it follows that our solution is a factor of at most $2$
away from $OPT$. 

Consider any arbitrary ordering of vertices $v_1, v_2, \ldots$
in $Q^*$. Note that the weight of each vertex in $Q^*$ is at most
$w(H^*)$. 
Let $Q^*_r$ denote the the first $r$ vertices in the above
ordering of vertices of $Q^*$. Let $p$ be the first index such that
$deg_{\overline{H}^*}(Q^*_p) > deg_{\overline{H}^*}(Q^*)/2$. This
implies the following -- (i) $deg_{\overline{H}^*}(Q^*_{p-1})\leq
deg_{\overline{H}^*}(Q^*)/2$, and (ii)
$deg_{\overline{H}^*}(Q^*\setminus Q^*_p)<
deg_{\overline{H}^*}(Q^*)/2$. Note that both the sets $Q^*_{p-1}$ and
$Q^*\setminus Q^*_p$ (neither set contains $v_p$) are feasible candidates for the set $Q$, the solution
returned by our algorithm when the heaviest vertex set was chosen to
be $H^*$. Since $w(Q)\geq w(Q^*_{p-1})$, $w(Q)\geq w(Q^*\setminus Q^*_{p})$, and
$w(v_p)\leq w(H^*)$, we have
\begin{eqnarray*}
w(OPT) & \leq w(H^*\cup Q^*) \\
       & \leq w(H^*) + w(Q^*) \\
       & \leq w(H^*) + w(Q^*_{p-1}) + w(v_p) + w(Q^*\setminus
Q^*_{p})\\
       & \leq w(H^*) + w(Q) + w(H^*) + w(Q) \\
       & = 2 w(H^* \cup Q) \\
       & \leq 2w(H'\cup Q')
\end{eqnarray*}
\end{proof}

\section{A 2-approximation for Fixed Weight Minimum Edge Cover}
Recall the \kmin problem: Given a graph $G=(V,E)$
with arbitrary vertex weights and a positive integer $W$, our objective
is to choose a set $S\subseteq V$ of vertices of total weight at least
$W$ such that that the number of edges with at least one end point in
$S$ is minimized. 

We will solve the following related problem optimally and then show
that an optimal solution to the problem is a 2-approximation to
\kmin: we want to find a subset $S$
of vertices such that $deg(S)$ is smallest and $w(S)$ is at least
$W$. 

We use the dynamic programming algorithm of the well-known Knapsack
problem to find a solution to the above problem. For completeness, we
restate the dynamic programming formulation below.
\begin{quote}
$P[i,D]$: maximum weight of set $Q\subseteq \{v_1,v_2,\ldots,v_i\}$
  such that $deg(Q)$ is at most $D$.   
\end{quote}
Note that $P[0,D]=0$, for all values of $D$ is the base case. For all
other case, we invoke the following recurrence.
\[
P[i,D] = \max\{P[i-1,D], w(v_i)+P[i-1,D-w(v_i)]\}
\]
After filling the table $P$ using dynamic programming, we scan all
entries of the form $P[|V|,D]$ to find the smallest value of $D$ for
which $P[|V|, D]\geq W$. Let $S$ be the corresponding set. 

\begin{lemma}
The
is a $2$-approximate solution to the Fixed Cost Minimum Edge Cover
Problem as follows.
\[
t(S) \leq deg(S) \leq deg(OPT) = 2(deg(OPT)/2)\leq 2OPT
\]
\end{lemma}

\section{Integrality gap for Fixed Cost Minimum Edge Cover}
Consider the following natural integer linear program for the problem
% and its LP relaxation.
\begin{eqnarray*}
\min \; \sum_{e} y_e\\
\mbox{subject\ to }\quad \quad \sum_{v\in V} x_v &\ge k,\ & \\
  y_e &\ge x_u,\ &\forall e = (u,v)\\
  y_e &\ge x_v,\ & \forall e = (u,v)\\
x_v &\in \{0,1\},\ &\forall v \in V \\
y_e & \in \{0,1\},\ &\forall e\in E
\end{eqnarray*}
The LP relaxation can be obtained by relaxing the integrality
constraints on $x_v$ and $y_e$ to $x_v \geq 0, \forall v \in V$ and 
$y_e \geq 0, \forall e\in E$.

\begin{theorem}
The above LP has an integrality gap of $2(1-o(1))$.
\end{theorem}
Let $k=\lfloor \sqrt{n}\rfloor$. 
Construct a graph $G$ on $n$ vertices as follows. For each pair of
vertices, include an edge between the pair with a probability
$1/\lfloor \sqrt{n}\rfloor$. For any vertex $v$, 
$E[deg(v)] = n(1/\lfloor \sqrt{n}\rfloor )\leq \lceil \sqrt{n}\rceil$. 
Using Chernoff bounds, for $0<\delta < 1$, we have
\[
\sqrt{n}(1-o(1)) \leq deg(v) \leq \sqrt{n}(1+o(1))
\] 
Consider any subset $Q$ of vertices in $G$ such that
$|Q|=\lfloor \sqrt{n}\rfloor $. Then we have 
\[
E[e(Q)] = \frac{1}{\lfloor \sqrt{n}\rfloor }{Q\choose 2} 
         = \frac{\lfloor \sqrt{n}\rfloor (\lfloor \sqrt{n}\rfloor
  -1)}{2\lfloor \sqrt{n}\rfloor }  
          = \frac{\lfloor \sqrt{n}\rfloor -1}{2}
\]
Thus, $n\geq 4$, we have $\sqrt{n}/4\leq E[e(Q)] < \sqrt{n}/2$. 
We use the following Chernoff bound to obtain the probability that $e(Q)\geq n^{1-\epsilon}$,
for a constant $\epsilon$.
\[
\Pr[e(Q)\geq (1+\delta)E[e(Q)]] \leq \left (\frac{\exp(\delta)}{(1+\delta)^{(1+\delta)}} \right )^{E[e(Q)]}
\]  
In our case, $2n^{1/2-\epsilon}\leq 1+\delta \leq  4n^{1/2-\epsilon}$, thus we get
\[
\Pr[e(Q)\geq n^{1-\epsilon}] \leq \left (\frac{\exp(4n^{1/2-\epsilon})}{(2n^{1/2-\epsilon})^{2n^{1/2-\epsilon}}} \right )^{\sqrt{n}/4}
\]
Let $f(n,\epsilon)= \left (
\frac{\exp(n^{1/2-\epsilon})}{(2n^{1/2-\epsilon})^{(n^{1/2-\epsilon}/2)}}
\right )^{\sqrt{n}}$.
The number of sets of size $\lfloor \sqrt{n}\rfloor$ is given by ${n\choose \sqrt{n}}\leq
(ne/\lfloor \sqrt{n}\rfloor )^{\sqrt{n}} = (\lceil \sqrt{n}\rceil
e)^{\sqrt{n}}$. The probability that there is no 
subset of size $\lfloor \sqrt{n}\rfloor$ that has at least $n^{1-\epsilon}$ edges is
given by the union-bound as follows
\[
f(n,\epsilon){n\choose \sqrt{n}} << 1
\]
The number of edges with at least one end point in $Q$ is given by
\begin{eqnarray*}
t(Q) & = deg(Q) - e(Q) \\
& \geq \lfloor \sqrt{n}\rfloor \cdot \sqrt{n}(1-o(1)) -
n^{1-\epsilon}\\
& = n(1-o(1))
\end{eqnarray*}
On the other hand, consider the fractional solution in which
$x_v=1/\sqrt{n}$, for each $v$ and $y_e=1/\sqrt{n}$, for each $e\in E$. This
LP solution is feasible and has a cost of $|E|/\sqrt{n}$. The number
of edges $|E| = n\sqrt{n}/2(1+o(1))$. Thus the cost of the LP solution
is at most $n(1+o(1))/2$, which results in a gap of $2(1-o(1))$.

\section{APX-hardness for unweighted Fixed Cost Minimum Edge Cover and
  Maximum Weight $m'$-Edge Cover}

Let $G$ be the input for the \dks problem and let $OPT$ be an optimal
subset of $k$ vertices. To prove the hardness result we consider the
following \textit{important} instance $\langle G,k\rangle$ of the \dks problem.

\begin{enumerate}
\item [$P_1$.] The $k/2$ largest degree vertices, $H$,  
in $G$, have average degree $d_H=\Theta(n^{1/3})$
\item [$P_2$.]
$k=\Theta(n^{2/3})$, and 
\item [$P_3$.]
$OPT$ has average degree $d^*=\Theta(n^{1/3})$.
\end{enumerate}

Feige et al. \cite{FKP} gave a relatively simple 
$n^{1/3}$ ratio for the {\dks} problem. 
The ratio was improved to $n^{1/3-1/60}$
by improving the ratio of $n^{1/3}$ for two very specific instances.
One of the instances was the important instance defined above.
We now show that if \kmin admits a PTAS then this important instance for
the \dks problem admits a constant factor approximation algorithm.

Consider the important instance and assume that $e(H)=O(k)$. Note that
a constant-factor approximation for the important instance implies a
constant approximation for the case when $e(H)=O(k)$. Clearly,
removing $H$ does not change the value of the optimum up to lower order
terms. This \textit{modified} instance has a maximum degree of $O(n^{1/3})$ and it
also satisfies properties $P_2$ and $P_3$ of the important instance.
The best ratio, given the state-of-the-art, for the modified instance is
$\Theta(n^{2/9})$ (\textit{M. Charikar, Private Communication}) and hence
the following conjecture seems 
highly likely: The modified instance does not admit a constant
approximation. 

\begin{claim}
We can modify $G$ into a graph $G'$ 
for which the optimal solutions for the \dks problem and the \kmin
are the same, and in addition, the value of the optimum solution 
for the \dks problem does not change.
\end{claim}
\begin{proof}
Let the largest degree of $G$ 
be $\Delta=c_1\cdot n^{1/3}$.

We show how to make the graph $\Delta$ regular without 
changing the optimum value for the \dks instance.
For every vertex $v\in V$ add 
$\Delta-deg(v)$ vertices $F_v$ and connect 
$v$ to all the vertices of $F_v$. 
The sets $F_v$ for different vertices are disjoint.
Let $F=\bigcup_{v} F_v$.
We now add a set of $n^2$ disjoint edges (no two edges share a vertex)
$F'$ (and thus $2n^2$ new vertices).
We then make $F\cup F'$ regular by adding 
adding a random $\Delta-1$ regular graph on $F'\cup F$.
Let $G'$ be the new graph.

Clearly, every vertex has degree $\Delta$ now.
Indeed all vertices in $F_v$ and the $2n^2$ vertices that were added had degree
exactly one before the random $\Delta-1$ subgraph is added.
Since $G'$ is regular, the sum of degrees 
in $G'$ for any $k$ vertices is the same.
As $t(U)=deg(U)-e(U)$ the optimal solutions for \kmin and \dks 
are the same on regular graphs.
Since the graph on $F\cup F'$ is basically a random 
graph with degrees $O(n^{1/3})$, and at least $n^2$ vertices, 
basic calculations show that for 
every $F''\subseteq F\cup F'$, $e(F'')=O(|F''|)$.
In addition, every vertex in $F'\cup F$ has degree at most $1$
in $G$. 
Therefore any $F''\subseteq F\cup F'$, can contribute at most 
$deg(F'')=O(|F''|)$ to the number of edges 
in a \dks solution.
As $|F''|\leq k$, it follows that $F''$ can contribute $O(k)$ edges to the 
the \dks solution.

Observe that this number is negligible compared to the \dks in $G$.
The number of edges in the \dks optimum in $G$ is
$c'kn^{1/3}$, for some constant $c'$.
Hence the value of the optimum solution does not change (up to lower
order terms)
by the addition of $F\cup F'$. 
\end{proof}

\begin{theorem}
\label{thm:kminhard}
A PTAS for \kmin problem implies a constant factor approximation for
the modified \dks instance.
\end{theorem}

\begin{proof}
Since $G'$ is a regular graph, 
the optimal \kmin solution 
is the same as the \dks optimum solution.
In fact the number of touching edges is
$\Delta k-c'kn^{1/3}$. 
Recall that $\Delta=c_1\cdot n^{1/3}$.
Thus the optimum is 
$c_1kn^{1/3}-c'kn^{1/3}$.

By the value of $k$, the optimum value is 
$c_1n-c'n$.
If \kmin has a PTAS then there exists a $1+c'/c_1$-approximation for
the \kmin problem. 
Assuming this ratio, a set $U$ is output so that it is touched
by at most 
$c_1n-c'n+(c'/c_1)(c_1n-c'n)
=c_1n-({c'}^2/c_1)n$ edges. This implies that $e(U)= (c'^2/c_1)n$.
Thus we find 
a subgraph with $k$ vertices and at least 
${c'}^2n/c_1$ edges. Therefore the ratio obtained on the modified instance
is $c'/(c'^2/c_1) = c_1/c'$, contradicting the conjecture that the modified
instance does not admit a constant approximation.
\end{proof}

\subsection{APX-hardness for Maximum Weight $m'$-Edge Cover}

We show that 
PTAS for (unweighted) \mwms implies a PTAS
for (unweighted) \kmin on the modified instance.
As we showed that this is not possible, 
\mwms is APX-hard as well.

Recall that the optimum for the modified instance 
had $c_1n-c'n$ edges and size $k$.
Furthermore, the modified instance is 
$\Delta$-regular.

Let $OPT$ be the optimum solution for the \kmin instance.
The number of edges touching $OPT$ is at least:
$t(OPT)\geq 
k\cdot \Delta/2.$ 
We impose a bound of $c_1n-c'n$ on the number of edges 
to the hypothetical PTAS for the \mwms 
problem.
A PTAS algorithm for \mwms will return a set $S$ with size {\em at least}
$k/(1+\epsilon)$, touched by 
at most $c_1n-c'n$ edges.

The amount of vertices still required to be added 
to transform $S$ to a legal \kmin output is 
$k-k/(1+\epsilon)=\epsilon\cdot k/(1+\epsilon)$.
We can complete the set $S$ to size $k$ by {\em any}
set $S'$ of 
$\epsilon\cdot k/(1+\epsilon)$ vertices.
In such case $t(S')\leq 
\epsilon\cdot \Delta\cdot k/(1+\epsilon)$.
As we showed before that $t(OPT)\geq k\cdot \Delta/2$,
it follows that 
$t(S\cup S')\leq t(OPT)+2\epsilon \cdot t(OPT)$

Thus for getting a ratio of $1+\delta$ for any constant $\delta$ just set $\epsilon=\delta/2$. Therefore, the assumption that the \mwms
problem admits a PTAS, implies that the \kmin problem 
admits a PTAS on the modified instance, which is highly unlikely, by
Theorem \ref{thm:kminhard}.

\section{An approximation for Fixed Cost Minimum Edge Cover implies the same approximation for Maximum Weight $m'$-Edge Cover}
We first transform the input instance for the \mwms problem to one in
which the optimum value of the objective function is  
at most $n^5$ by paying a very small penalty in the approximation ratio.

\begin{lemma}
\label{lemma:newInstance}
For the Maximum weight $m'$-subgraph problem, we can convert the input
instance $\langle G,w,m'\rangle$, with an optimal
solution denoted by $OPT$ into an instance $\langle 
G',w',m'\rangle$, with optimal solution OPT'', such that $OPT''\leq
n^5$. Furthermore, if $OPT'$ is the total weight of the vertices in
$OPT''$ under the weight function $w$, then 
\[
OPT'\geq OPT(1-1/n)(1-1/n^2)
\]  
\end{lemma}

\begin{proof}
Let $v_1, v_2, \ldots, v_n$ be the vertices in $G$ such that 
$w(v_1)\geq w(v_2)\geq \cdots \geq w(v_n)$. Let $v_p$ be the last
vertex in the ordering such that $w(v_p)\geq w(v_1)/n^2$. In other
words, for each $j$, $p<j\leq n$, $w(v_1)> n^2w(v_j)$. Let $G'$ is the
graph induced on vertices $v_1,v_2, \ldots, v_p$. Let OPT$_1$ be the
optimal solution for the instance $\langle G',w,m'\rangle$. Note that
OPT may choose some vertices from the set $\{v_{p+1}, v_{p+1}, \ldots,
v_{n}\}$. The error incurred in not considering these vertices is at
most $n(w(v_1)/n^2) \leq OPT/n$. Thus we get
\[
OPT_1 \geq OPT(1-1/n)
\]
We now scale the weights of vertices in $G'$ to create an instance
$\langle G',w',m'\rangle$, where 
\[
w'(v_j) = \left \lfloor \left (\frac{w(v_j)}{w(v_p)}\right )n^2  \right \rfloor
\]
Let $OPT''$ be an optimal solution to $\langle
G',w',m'\rangle$. Clearly, $OPT''\leq n^5$.
Let $OPT'$ be the cost of the solution
$OPT''$ under the weight function $w$, i.e., $OPT' =
\sum_{v\in OPT''}w(v)$. Thus we have
\begin{equation}
OPT' \geq OPT_1\left (1 - \frac{1}{n^2}\right )
\geq OPT\left (1 - \frac{1}{n}\right )\left (1 - \frac{1}{n^2}\right )
\end{equation}
\end{proof}

\begin{theorem}
\label{thm:}
For some constant $\alpha$, an $\alpha$ approximation guarantee for
{\kmin} implies an $\alpha(1+o(1))$ approximation guarantee for {\mwms}. 
\end{theorem}
\begin{proof}
Suppose that we have an $\alpha>1$ approximation algorithm 
for {\kmin}, for some constant $\alpha$.
Using Lemma \ref{lemma:newInstance}, we transform the {\mwms} instance
$(G,m')$ with an optimal weight $W^*$ to an instance in which the optimum weight 
$W^*\leq n^4$. This increase the approximation ratio by a factor of
only $(1+o(1))$.
We now consider the modified instance $(G',m')$ as an 
input to {\kmin}. We guess the value of $W^*$ by trying all possible
integral values between $1$ and $n^4$. For each guess of $W^*$, we
apply the $\alpha$-approximation algorithm for {\kmin} to the new
instance. 
When our guess $W^*$ is correct and we apply the algorithm, we obtain a
set $U$ of vertices of cost at least $W^*$ and that touch at most 
$\alpha\cdot m'$ edges.

Create a new set $B$ in which every vertex from $U$ is chosen 
with a probability $1/\alpha$.
We say that an edge $e$ is {\em deleted} if $e\not\in E(B)$.
Let $\tau$ be a constant.

We consider the following ''bad" events: (i) $w(B)\leq
W^*/((1+\tau)\alpha)$, (ii) $t(B)>m'$.

We first bound the probability that $w(B)\leq W^*/((1+\tau)\alpha)$.
The expected cost of $B$ is $w(U)/\alpha= W^*/\alpha$.
Consider the expected cost of $U\setminus B$. The expected cost is $W^*- W^*/\alpha$.
The event that $w(B)\leq W^*/(\alpha(1+\tau)))$
is equivalent to the event 
$w(U)-w(B)\geq W^*-W^*/(\alpha(1+\tau))=W^*(1-1/(\rho(1+\tau))$.
By the Markov's inequality, the last event has probability at most
$(1-1/\alpha)/(1- 1/(\alpha(1+\tau))= 1-\tau/(\alpha+\alpha\cdot \tau-1)$

We now bound the probability of the second bad event.
The expected number of edges in $E(B)$ is
at most $m'(1-(1-\frac{1}{\alpha})^2)$.  
Note that the events that edges are deleted 
are positively correlated because given 
that an edge $(v,u)$ is deleted, one of the possibilities that 
can cause this event, is that $v$ is deleted, and in that case all
edges of $v$ are deleted with probability $1$. Clearly, we can assume
that $m'\geq c$ for any constant $c$. Otherwise, we can solve the
{\mwms} problem in polynomial time by checking all subsets of edges.
By the Chernoff bound, the probability that the number of edges is
more than $m'$ is bounded by $exp(-c\delta^2/2)$, for some $\delta <
1$. We can choose a large enough $c$ so that the above probability is
at most $\tau/(2(\alpha+\alpha\cdot \tau -1))$.
This would mean that the sum of probabilities of bad events is strictly 
smaller than $1$.
This construction can be derandomized 
by the method of conditional expectations. 
\end{proof}

\section{Exact algorithm for the Degrees Density Augmentation Problem}
The \dap problem is as follows: Given
a graph $G=(V,E)$ and a subset $U\subseteq V$, the objective is to
find a subset $W\subseteq V\setminus U$ such that 
\[
\rho = \frac{e(W)+e(U, W)}{deg(W)} \mbox{ is maximized}
\]

The \dap problem is related to the \kmin problem in the same way as the
Densest subgraph problem is related to the Dense $k$-subgraph
problem.  A natural heuristic for the \kmin problem would be to
iteratively find a set $W$ with good augmentation degrees density. 
A polynomial time exact solution for the problem using linear
programming is given in \cite{CMNRS}. Here we present a combinatorial
algorithm. 
 
We solve the \dap problem exactly by finding minimum $s$-$t$ cut in
the flow network constructed as follows. Let $\overline{U}$ denote the
set $V\setminus U$. In addition to the source $s$
and the sink $t$, the vertex set contains $V_{E'}\cup \overline{U}$,
where $V_{E'}=\{v_e\, |\, e\in E \mbox{ and both end points of $e$ are
in } \overline{U}\}$.
There is an edge from $s$ to every vertex in $V_{E'}\cup \overline{U}$. If
$a$ is a vertex in $V_{E'}$ then the capacity of the edge
$(s,a)$ is $1$, otherwise, the capacity of the edge is
$deg_U(a)$.  For each vertex $v_e\in V_{E'}$, where $e=(p,q)$, there
are edges $(v_e,p)$ and $(v_e,q)$. Each such 
edge has a large capacity of $M= \infty$ (any capacity of at least
$n^5$ would work).
Finally, each
vertex $p\in \overline{U}$ is connected to $t$ and has a capacity of
$\rho\cdot deg(p)$. 
%The construction is depicted in the figure below.
%
%begin{figure}[h]
%centering
%includegraphics[width=3.7in]{flow-crop.pdf}
%caption{Flow Network.} 
%label{fig:flow}
%end{figure}

\subsection{Algorithm}
For a particular value of $\rho$, let $W_s\subseteq \overline U$ be
the vertices that are on the $s$($t$) side of a minimum $s$-$t$
cut. Let $V_{E'}^s\subseteq V_{E'}$($V_{E'}^t\subseteq V_{E'}$) be the
vertices in $V_{E'}$ that are on the $s$($t$) side of the minimum $s$-$t$
cut. We now state the algorithm.

\begin{enumerate}
\item[1.]
Construct the flow network as shown above.
\item[2.]
For each value of $\rho$, compute a minimum $s$-$t$ cut and find the
resulting value of 
$e(W_s) + e(U, W_s) - \rho deg(W_s)$. Find the largest value of
$\rho$ for which the expression is at least $0$. 
\item[3.]
Return $W_s$ corresponding to the largest value of $\rho$.
\end{enumerate}

\subsection{Analysis}
\begin{lemma}
\label{lemma:mincut}
Any minimum $s$-$t$ cut in the above flow network has 
capacity at most $2n^2$.
\end{lemma}

\begin{proof}
This follows because the $s$-$t$ cut $(s,V\setminus \{s\})$ has
capacity at most $2n^2$.
\end{proof}

\begin{lemma}
\label{lemma:allins}
For any minimum $s$-$t$ cut $C$, $|V_{E'}^s| = e(W_s)$. 
\end{lemma}

\begin{proof}
Note that it cannot be the case that $|V_{E'}^s| > e(W_s)$, as this
will result in the
capacity of the cut $C$ being at least $M$, which is not possible by
Lemma \ref{lemma:mincut}. Note that
any $s$-$t$ cut for which $|V_{E'}^s| < e(W_s)$ can be
transformed into another $s$-$t$ cut of a lower capacity in which $|V_{E'}^s| =
e(W_s)$ by moving vertices in $V_{E'}^t$ that correspond to edges in
$W_s$ to $V_{E'}^s$. Since edges from $s$ to vertices in $V_{E'}$
(vertices in $V_{E'}^t$, in particular) have
capacity of $1$, 
the capacity of the cut reduces. The claim follows. 
\end{proof}

\begin{lemma}
The \textsf{Degrees Density Augmentation} problem admits a polynomial
time exact solution.
\end{lemma}

\begin{proof}
We are interested in finding a non-empty set $W_s\subseteq
\overline{U}$ such that  
\(
\frac{e(W_s) + e(U,W_s)}{deg(W_s)} \mbox{ is maximized }. 
\)
Note that there are at most $2n^4$ possible values of $\rho$ that our
algorithm needs to try.
Indeed, the numerator is an integer between $1$ and $2n^2$
and the denominator is an integer between $1$ and $n^2$. 
Since minimum $s$-$t$ cut can be computed in polynomial time, our
algorithm runs in polynomial time. 
 
For any fixed guess for $\rho$, the capacity of the min $s$-$t$ cut is
given by
\begin{eqnarray*}
& \min_{W_s\subseteq \overline{U}} {|V_{E'}^t|  + deg_U(W_t) + \rho
deg(W_s)} \\
= & \min_{W_s\subseteq \overline{U}} {|V_{E'}|-|V_{E'}^s| +
deg_U(\overline{U}) - deg_U(W_s) + \rho deg(W_s)} \\
= & |V_{E'}| + deg_U(\overline{U}) - \max_{W_s\subseteq \overline{U}}
{|V_{E'}^s| + deg_U(W_S) -\rho deg(W_s)}\\
= & |V_{E'}| + deg_U(V\setminus U) - \max_{W_s\subseteq \overline{U}}
{e(W_s) + e(U,W_S) -\rho deg(W_s)} (\mbox{ using Lemma \ref{lemma:allins}})\\
\end{eqnarray*}
Our algorithm ensures that $\rho deg(W_s) \geq e(W_s) + e(U,W_S)$,
which eliminates the possibility of $W_s=\emptyset$. Thus, finding the
minimum $s$-$t$ cut for a fixed $\rho$ in the above flow network is
equivalent to finding a set $W_s$ with the largest degree
density. Thus we have
\[
\frac{e(W_s) + e(U,W_s)}{deg(W_s)} \geq \rho
\]
Since our algorithm finds such $W_s$ for each possible fraction that
$\rho$ can assume and returns the $W_s$ with the highest degree
density, our solution is optimal.
\end{proof}

\subparagraph*{Acknowledgements}
We thank V.\ Chakravarthy for introducing the \kmin problem to us. We
also thank V.\ Chakravarthy and S.\ Roy for useful discussions.

\bibliographystyle{abbrv}
%\bibliography{u1}

\end{document}